\documentclass[a4paper]{ifacconf}

\usepackage{amsmath, amssymb, amsfonts}
\usepackage{bbm}
\allowdisplaybreaks
\usepackage{graphicx}
\usepackage{natbib}
\usepackage{subcaption}
\usepackage{tikz}
\usetikzlibrary{arrows.meta}
\usepackage{xcolor}

\newcommand{\qedsymbol}{$\blacksquare$}
\definecolor{dgreen}{rgb}{0.,0.8,0.}
\newcommand{\green}[1]{{\color{black}#1}}

\newtheorem{theorem}{Theorem}
\newtheorem{lemma}[theorem]{Lemma}
\newtheorem{proposition}[theorem]{Proposition}
\newtheorem{corollary}[theorem]{Corollary}
\theoremstyle{definition}

\newtheorem{assumption}[theorem]{Assumption}
\theoremstyle{remark}
\newtheorem{remark}[theorem]{Remark}

\newenvironment{proof}[1][Proof]{%
  \par\noindent\textbf{#1. }\rmfamily
}{\hfill\qedsymbol\par}

\newcommand*{\dif}{\mathop{}\!\mathrm{d}}
\newcommand{\reals}{\mathbb{R}}
\newcommand{\naturals}{\mathbb{N}}
\newcommand{\D}{\mathbb{D}}
\newcommand{\diag}{\mathrm{diag}}

\newcommand{\sbs}[1]{_{\scriptscriptstyle \mathrm{#1}}}

\newcounter{example}
\renewcommand{\theexample}{\arabic{example}}

\newenvironment{example}[1][]{%
  \refstepcounter{example}%
  \par\noindent
  \textbf{Example~\theexample%
  \if\relax\detokenize{#1}\relax
    .%
  \else
    \, -- #1.%
  \fi}%
  \rmfamily
}{\hfill\qedsymbol\par}

\begin{document}
\begin{frontmatter}

\title{Pathological Regimes of Closed-Loop Recommendation Systems over Social Networks \thanksref{footnoteinfo}}

\thanks[footnoteinfo]{This work has been supported in part by ANR projects Feeding Bias (ANR-22-CE380017-01) and MIAI Cluster BEAR (ANR-23-IACL-0006).}

\author[First,Second]{Simone Mariano}
\author[First]{Paolo Frasca}

\address[First]{Univ.\ Grenoble Alpes, CNRS, Inria, Grenoble INP, GIPSA-lab, 38000 Grenoble, France (e-mail: simone.mariano@grenoble-inp.fr, paolo.frasca@gipsa-lab.fr).}
\address[Second]{Univ.\ Grenoble Alpes, CNRS, Sciences Po Grenoble-UGA, Pacte, 38000 Grenoble, France.}

\begin{abstract}
This paper addresses the problem of designing recommendation systems for social networks and e-commerce platforms from a control-theoretic perspective. We formulate recommendation design as an infinite-horizon state-feedback optimal control problem whose performance index rewards alignment/engagement while penalizing polarization, large deviations from an uncontrolled baseline, recommendation mismatch, control effort, and exposure disagreement across neighboring users.
We derive explicit spectral conditions under which the reduced quadratic stage cost is strictly positive-definite, and we show that the failure of these conditions makes the resulting recommendation design exhibit pathological behaviors, such as unstable free modes, non-attainment of the infimum, or failure of the stationary affine synthesis.
\end{abstract}


\end{frontmatter}

\section{Introduction}
\label{sec:intro}

Recommendation systems increasingly mediate how users encounter information, products, and social content in online spaces, such as social media and e-commerce platforms. While recommendations are often optimized to improve short-term engagement, the resulting feedback loop between recommendations and user behavior may amplify confirmation bias and contribute to undesirable collective phenomena, including polarization, echo chambers, and extreme opinions \citep{gausen2022abm,huszar2022algorithmic,bail2018opposing}. Moreover, combined with insufficient exploration, maximizing short-term engagement may lead to preferences to drift toward degenerate or extreme contents~\citep{jiang2019degenerate}, and repeatedly training on self-influenced logs may reduce diversity and long-run utility \citep{chaney2018confounding}. A basic mechanism behind these effects is that engagement-oriented recommendations can push users toward more extreme or polarized positions, while extreme positions may themselves generate higher engagement and monopolize users' attention. This mechanism is captured in closed-loop form in the single-user model by~\cite{FrascaRecommendation2022}, which provides a useful baseline for studying recommendation-induced opinion dynamics. 

The perspective developed by~\cite{dean2024accounting}  highlights the limitations discussed above and argues that recommendation systems should be designed with explicit models of how users and algorithms shape one another. This point is particularly relevant because several common design choices, including memoryless architectures~\citep{covington2016youtube}, simplified models of user dynamics, omitted creator adaptation, and aggressive engagement-oriented optimization~\citep{chen2019topk,immorlica2024clickbait}, can obscure long-term causal effects. Usually ad hoc, \emph{a posteriori} mitigation strategies for ill-behaved systems are often introduced only after deployment, using logged data generated by the same platform dynamics they are meant to correct, and may fail to prevent exposure shifts, reduced
content diversity, feedback-loop amplification, and polarization \citep{sinha2016deconvolving,chaney2018confounding,mansoury2020feedback,nguyen2014exploring}. 

In contrast, the recent survey paper
\cite{depasquale2026recommender} presents a control-theoretic perspective, in which recommender systems are interpreted as feedback systems and fairness-related objectives are analyzed as long-term dynamical properties. These observations suggest that recommendation design should not be treated only as a prediction or ranking problem, but also as a
dynamical feedback problem. The present paper studies one specific aspect of this broader agenda. We ask when an optimal-control formulation of recommendation design yields a well-posed closed-loop system, and when the same formulation may instead produce pathological opinion dynamics. The user side of the platform is modeled as a networked population holding opinions on multiple coupled topics. The recommendations are modeled as control inputs acting on a continuous-time opinion dynamics model inspired by~\cite{friedkin2016network,ye2020continuous}. The social influence among users is represented by a graph Laplacian, the logical dependence among topics is encoded by a topic-coupling matrix, and anchoring terms describe the tendency of users to retain their inner beliefs.                  
Building upon this modeling effort, 
the recommendation policy is synthesized through an infinite-horizon state-feedback optimal control problem. The corresponding performance index makes the platform's design trade-offs explicit, as it rewards alignment between recommendations and user opinions, used here as a proxy for engagement, while penalizing polarization, deviation from the uncontrolled equilibrium, recommendation effort, disagreement of exposure across neighboring users, recommendation mismatch, and the distance between the recommended stance and the currently expressed opinion. 

The main point of the paper is that these weights do more than tune a desired performance criterion. Indeed, they determine the mathematical regime of the resulting closed-loop design problem. When the reduced quadratic stage cost is strictly positive definite, the formulation falls within the classical infinite-horizon LQ setting and the usual well-posedness intuition applies. When this positivity is lost, however, the same closed-loop architecture may enter semidefinite or sign-indefinite regimes in which optimality of the prescribed cost no longer guarantees stability or even attainability of an optimal recommendation input.

We characterize these regimes in two steps. First, we analyze the homogeneous problem, obtained by removing the affine drift and linear cost terms. In this setting, we derive algebraic and spectral conditions that identify the positive-definite regime and then use free-endpoint indefinite LQ theory to describe what can happen outside it. This analysis distinguishes stabilizing behavior from unstable free modes and from cases in which the infimum of the cost is finite but not attained. Second, we reintroduce the affine terms induced by the uncontrolled opinion equilibrium and by the linear part of the cost. For this affine problem, we derive a stationary quadratic-affine Hamiltonian
identity and isolate the associated linear-algebraic compatibility condition. This condition determines whether the affine forcing can be absorbed by a constant feedback bias or whether the stationary affine synthesis fails.

A collection of significant examples shows that the resulting pathologies are not merely formal edge cases. Depending on the choice of weights, the optimal-control formulation may produce an unstable free mode, a finite but unattained optimum, semidefinite optimality without closed-loop stability, an affine translation of a homogeneous instability, or a failure of the affine compatibility condition. These examples make explicit how an apparently reasonable engagement-oriented objective can lead to a mathematically ill-posed or dynamically undesirable recommendation design when regularization, polarization penalties, and mismatch penalties are not sufficiently strong.

Compared with the conference version~\cite{mariano2026optimal}, this paper makes three main extensions. First, the performance index is extended by inclusion of the recommendation-mismatch term $J_{\mathrm C}$, yielding a more general reduced LQ formulation. Second, the analysis separates the homogeneous free-endpoint backbone from the affine layer generated by the baseline drift and the linear cost term. Third, the examples are expanded to provide a classification of pathological closed-loop recommendation designs, including both homogeneous indefinite-LQ failures and affine compatibility obstructions.

The remainder of the paper is organized as follows. Section~\ref{sec:pi} introduces the performance index and formulates the infinite-horizon optimal control problem. Section~\ref{sec:model} presents the networked multi-topic opinion dynamics and derives the reduced LQ formulation. Section~\ref{sec:sols_homo} studies the homogeneous problem and distinguishes positive-definite, semidefinite, and sign-indefinite regimes. Section~\ref{sec:sols_affine} treats the affine extension and derives the
stationary compatibility condition. Section~\ref{sec:problems} presents the possible pathological behaviors. Section~\ref{sec:disc} concludes with a discussion of the modeling choices, limitations, and possible extensions.

\noindent\textbf{Notation:} $\reals$ and $\naturals$ denote the sets of real and natural numbers, respectively.
For $n\in\naturals$, $\mathbbm{1}_n$ is the $n$-dimensional vector of ones, and
$[n]:=\{1,\dots,n\}$. $I_n$ and $O_n$, $n\in\naturals$, denote identity and zero matrices of appropriate dimensions. For any vector $x\in\reals^n$, $(x)=x^\top$. For $A,B \in \reals^{n \times m}$, $A \odot B$ denotes their Hadamard (elementwise) product, i.e., $(A \odot B)_{ij} := A_{ij} B_{ij}$ for all $i\in [n]$ and $j\in [m]$.
For symmetric $A,B\in\reals^{n \times n}$, the relations $A\succeq B$ and $A\preceq B$ denote the Loewner order, while $A \succ 0$ and $A \prec 0$ ($A \succeq 0$ and $A \preceq 0$) denote positive (semi)definite and negative (semi)definite matrices, respectively.
For $n\in\naturals$, $\mathbb{D}^n$ is the set of real diagonal matrices with $n$ diagonal entries, while
$\mathbb{D}_{\succ0}^n$ (resp.\ $\mathbb{D}_{\succeq0}^n$) is the set of diagonal positive
definite (positive semidefinite) matrices with $n$ diagonal entries.
For a diagonal matrix $W_{(\cdot)}=\mathrm{diag}(w_{(\cdot),1},\dots,w_{(\cdot),nm})$,
we write
$\lambda_{M,(\cdot)}:=\max_{i\in\{1,\dots,nm\}} w_{i,(\cdot)}$ and
$\lambda_{m,(\cdot)}:=\min_{i\in\{1,\dots,nm\}} w_{i,(\cdot)}$.
For a complex number $\lambda$, $\Re(\lambda)$ denotes its real part, and for a matrix $M$,
$\sigma(M)$ denotes its spectrum.

\section{Performance index and design goal}
\label{sec:pi}

The first objective of this work is to clearly define a performance index that quantifies the distortion and polarization-versus-engagement issue identified in \cite{FrascaRecommendation2022} in a networked, multi-topic setting. The guiding principle in designing the recommendation system is simple: engagement should be rewarded, but only to the extent that it does not generate pathological dynamics.

Consider a set of $n \in \naturals$ users holding opinions on $m \in \naturals$ topics, connected via a directed, weighted and connected graph $\mathcal{G}(\mathcal{E},\mathcal{V})$, where $\mathcal{E}$ and $\mathcal{V}$ are, respectively, the edge and vertex set, and with associated Laplacian matrix $L\in \reals^{n\times n}$, with $L=\Delta-\Gamma$, where $\Gamma\in\mathbb{R}^{n\times n}$ is the adjacency matrix and $\Delta= \diag(\Gamma \mathrm \mathbf{1}_n) $ is the degree matrix. For any time $t\ge 0$, let $x(t)\in \reals^{nm}$ be the vector of opinions and $u(t)\in \reals^{nm}$ be the vector of the inputs provided by the recommendation systems. The entries $x_{(k-1)n + i}$ and $u_{(k-1)n + i}$, $k\in\{1, \dots,m\}$, $i\in\{1, \dots,n\}$, of $x$ and $u$ correspond to the opinion of the $i$-th agent on the $k$-th topic and its corresponding input, respectively.

The aforementioned trade-offs will be precisely quantified through the integral cost
\begin{align}
J(x(t),u(t))
:=& \int_{0}^{\infty}\big(-J\sbs{EN}+J\sbs{P}+J\sbs{D}+J\sbs{EX}+J\sbs{F}+{J\sbs{C}}\big)\,\dif t \nonumber\\
=&\int_{0}^{\infty} \ell(x(t),u(t))\,\dif t,
\label{eq:J}
\end{align}
whose different terms correspond to the various objectives of the recommendation system.

The first term
\begin{equation*}
J\sbs{EN}:=x^{\top} W\sbs{EN} u,
\end{equation*}
with $W\sbs{EN} \in \mathbb{D}_{\succeq0}^{nm}$, models the engagement by rewarding the alignment between user opinion and recommendation.
Notice that by convention the index is minimized, and hence engagement appears with a negative sign.

The term
\begin{equation*}
J\sbs{P}:=x^\top W\sbs{P}x,
\end{equation*}
with $W\sbs{P} \in \mathbb{D}_{\succeq0}^{nm}$, penalizes the polarization as in network-aware designs and agent-based evaluations~\citep{chandrasekaran2024network,gausen2022abm}.

The term
\begin{equation*}
J\sbs{D}:=(x-x_{\mathrm{eq}})^\top W\sbs{D}(x-x_{\mathrm{eq}}),
\end{equation*}
with $W\sbs{D} \in \mathbb{D}_{\succeq0}^{nm}$, captures how much opinions deviate from the uncontrolled equilibrium $x_{\mathrm{eq}}\in \reals^{nm}$, which generalizes the inner belief idea in~\cite{FrascaRecommendation2022} and \cite{friedkin2016network} and aims to preserve the alignment of the agents' inner beliefs and their expressed opinion on a given topic.

The term
\begin{equation*}
J\sbs{EX}:=u^\top W\sbs{EX}u,
\end{equation*}
with $W\sbs{EX} \in \mathbb{D}_{\succ0}^{nm}$, captures the effort of the control and has the scope of limiting excessively strong or frequent inputs to avoid overexposure of the agents in the social platform \citep{qin2024too}.

Finally, the term
\begin{equation*}
J\sbs{F}:=\alpha\sbs{F}\,u^\top L_u u,
\qquad
L_u:=I_m\otimes L_b, \qquad \alpha\sbs{F}\geq 0,
\end{equation*}
with $L_b:=\Delta_b-\Gamma_b$,  $\Gamma_b:=\tfrac{\Gamma+\Gamma^\top}{2}$ and $\Delta_b:=\operatorname{diag}(\Gamma_b\mathbf{1})$, is used to mimic collaborative filtering by regularizing exposure across neighboring users and robustifies the design without imposing hard constraints. 
This design choice is justified under the widely-supported assumption that users who interact with one another tend to share similar preferences and opinions \citep{mcpherson2001birds}.

Contrarily to \cite{mariano2026optimal}, we also include the recommendation-mismatch term
\begin{equation*}
J\sbs{C}:=(u-x)^\top W\sbs{C}(u-x),
\end{equation*}
with $W\sbs{C}\in\mathbb{D}_{\succ0}^{nm}$, which penalizes excessive discrepancy between the recommended stance and the currently expressed opinion. This term discourages recommendations that are too far from the user's present viewpoint and therefore complements the overexposure and polarization penalties by explicitly regularizing recommendation mismatch.

\begin{figure}[tbp]
  \centering
  \raisebox{0.75ex}{\scalebox{1.20}{\begin{tikzpicture}[
      >=Latex, line width=0.8pt, line join=miter, line cap=butt,
      every node/.style={font=\scriptsize},
      block/.style={draw, thick, align=center,
                    minimum width=2.2cm, minimum height=0.9cm, inner sep=2pt}
  ]
    \node[block] (user) at (0,1.4)   {User dynamics\\ $f(x(t),u(t))$};
    \node[block] (rec)  at (0,0)     {Recommendation system\\$\pi(x(t),u(t))$};

    \draw[->] (rec.north) -- node[right,pos=0.55] {$u$} (user.south);

    \path (user.east) ++(0, 0.20) coordinate (yout);
    \path (user.east) ++(0,-0.20) coordinate (xfb);

    \draw[->] (yout) -- ++(0.9,0) node[above,pos=0.15] {$\qquad x \odot u$};

    \draw[->] (xfb) -- ++(0.9,0) node[below,pos=0.30] {$x$}
                |- (rec.east);
  \end{tikzpicture}}}
  \caption{Representation of the user--recommendation-system feedback loop.}
  \label{fig:sys}
\end{figure}
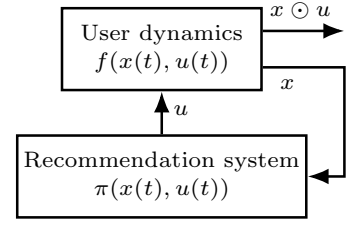

The design goal is to define a recommendation system that selects the appropriate $u$ so that the population follows a well-behaved trajectory and settles near a desirable steady regime, while minimizing~\eqref{eq:J}; see Figure~\ref{fig:sys}. This objective translates into the infinite-horizon optimal control problem
\begin{align}
   \min_{u}\quad &  J=\int_{0}^{\infty} \ell(x(t),u(t))\,\dif t, \nonumber\\
   \mathrm{s.t.}\quad & \dot x = f(x,u),\qquad x(0)=x_0.
\label{eq:prob_full}
\end{align}

\section{Opinion dynamics model}
\label{sec:model}

To define the agents' dynamics, we follow the models presented in \cite{friedkin2016network,ye2020continuous}. Let $X \in \reals^{n\times m}$ be the opinion matrix, where $X_{ik}$ denotes the opinion of agent $i\in\{1,\dots,n\}$ on topic $k\in\{1,\dots,m\}$. The continuous-time dynamics are
\begin{align}
    \dot{X}(t) = - L\,X(t) \;&-\; A_a\big(X(t)-X_\circ\big) \nonumber \\
    &+\; \big(U(t)-X(t)\big) \;+\; \big(X(t)C^\top - X(t)\big),
\label{eq:model_matrix_noninv}
\end{align}
with the single opinion evolving by
\begin{align*}
    \dot x_{ik}
:= & -\sum_{j=1}^{n}\big( L_{ij}\,x_{jk}(t)\big)
  - A_{a,ii}\big(x_{ik}(t)-X^\circ_{ik}\big) \\
  &+ (u_{ik}(t) - x_{ik}(t))
  + \Big(\sum_{h=1}^{m} x_{i h}(t)\,C_{kh} - x_{ik}(t)\Big).
\end{align*}
Matrix $L\in\reals^{n\times n}$ is the Laplacian of $\mathcal{G}(\mathcal{E},\mathcal{V})$, which drives consensus among neighboring agents (within each topic), $C\in\reals^{m\times m}$ captures inter-topic influence within each agent's opinions, while $A_a\in\D^{n\times n}_{\succ0}$ is a diagonal anchoring matrix with $X_\circ\in\reals^{n\times m}$ collecting anchoring opinions, namely the inner beliefs of the agents on a given topic. Finally, the matrix $U(t)\in\reals^{n\times m}$ is the input provided by the recommendation system and appears in the relative form $(U-X)$, so that aligned recommendations induce no artificial amplification.

The inter-topic matrix $C$ satisfies the following property, consistent with Assumption~1 in \cite{ye2020continuous}, which prevents instability of the uncontrolled system.

\begin{assumption}
\label{ass:C}
Matrix $C$ is such that $c_{ii} \geq 0$ for all $i \in [m]$, $|c_{ij}| \leq 1$ for all $i, j \in [m]$, and, given $A:=C-I_m$, zero is a semisimple eigenvalue of $A$ with multiplicity $p\geq 1$, while for every $\lambda\in\sigma(A)$ such that $\lambda \neq 0$ one has $\Re(\lambda)<0$.
\end{assumption}

Equation~\eqref{eq:model_matrix_noninv} extends \cite{friedkin2016network,ye2020continuous} by including a recommendation input $U$ that specifies a stance per user and topic. At the same time, it extends the single-user closed-loop model of \cite{FrascaRecommendation2022} to a networked, multi-topic, continuous-time setting.

We now provide a convenient vectorized form of \eqref{eq:model_matrix_noninv}.
Let $x:=\mathrm{vec}(X)\in\reals^{nm}$ and $u:=\mathrm{vec}(U)\in\reals^{nm}$. Using the identity $\mathrm{vec}(AXB)=(B^\top\otimes A)\mathrm{vec}(X)$, one obtains
\begin{align}
\dot{x}
:= &\Big[(C\otimes I_n) - (I_m\otimes(L+A_a)) - 2I_{nm}\Big]x \nonumber \\
&+\; (I_m\otimes A_a)\,\mathrm{vec}(X_\circ) \;+\; u
= A_c x + d + u
=: f(x,u),
\label{eq:aff_dyn}
\end{align}
with
\begin{equation}
A_c := (C\otimes I_n) - \big(I_m\otimes(L + A_a + 2I_n)\big),
\,
d := (I_m\otimes A_a)\,\mathrm{vec}(X_\circ).
\label{eq:model_vec_comp}
\end{equation}
The uncontrolled equilibrium solves $A_{uc}x_{\mathrm{eq}}+d=0$ with
\begin{equation*}
A_{uc} := (C\otimes I_n) - \big(I_m\otimes(L + A_a + I_n)\big),
\qquad
x_{\mathrm{eq}} := -A_{uc}^{-1}d,
\end{equation*}
where $A_{uc}$ is Hurwitz under Assumption~\ref{ass:C}; see Lemma~2 in \cite{ye2020continuous}. In particular, when $X_\circ$ is constant, the uncontrolled dynamics $\dot x = A_{uc}x + d$ is stable and converges to $x_{\mathrm{eq}}$.

The next result builds on Lemma~2 in \cite{ye2020continuous} and shows that $A_c$ is also Hurwitz.

\begin{proposition}
\label{prop:Ac_Hur}
Given Assumption~\ref{ass:C}, $A_c$ in \eqref{eq:model_vec_comp} is Hurwitz.
\end{proposition}

\begin{proof}
Let
\begin{equation*}
M:=L+A_a+2I_n,
\qquad
A_c=C\otimes I_n-I_m\otimes M.
\end{equation*}
By Assumption~\ref{ass:C}, if $\lambda\in\sigma(C)$, then $
\Re(\lambda)\le 1.
$
Indeed, $C=A+I_m$, where $A:=C-I_m$, and Assumption~\ref{ass:C} gives
$\Re(\mu)\le 0$ for every $\mu\in\sigma(A)$, with the only eigenvalues on the
imaginary axis equal to $0$.

Next, let $\mu\in\sigma(M)$. Since $L$ is a Laplacian, its off-diagonal entries
are nonpositive and
\begin{equation*}
\sum_{j\neq i}|L_{ij}|=L_{ii}.
\end{equation*}
Hence, by Gershgorin's theorem, every eigenvalue $\mu$ of $M$ satisfies
\begin{equation*}
\Re(\mu)\ge \min_{i\in[n]}\Big(M_{ii}-\sum_{j\neq i}|M_{ij}|\Big)
=\min_{i\in[n]}\big((A_a)_{ii}+2\big)
=:2+a_{\min},
\end{equation*}
where $a_{\min}:=\min_{i\in[n]}(A_a)_{ii}>0$.

Finally, by the standard spectral identity for Kronecker sums,
\begin{equation*}
\sigma(C\otimes I_n-I_m\otimes M)
=
\{\lambda-\mu:\ \lambda\in\sigma(C),\ \mu\in\sigma(M)\}.
\end{equation*}
Therefore, for every $\nu\in\sigma(A_c)$,
\begin{equation*}
\Re(\nu)=\Re(\lambda-\mu)\le 1-(2+a_{\min})=-(1+a_{\min})<0.
\end{equation*}
Thus $A_c$ is Hurwitz.
\end{proof}

To analyze the role of the different design terms and to separate the genuinely free-endpoint homogeneous backbone from the additional affine contributions induced by the baseline drift and linear terms, we first rewrite the recommendation problem in an equivalent reduced LQ form. This reformulation makes explicit the matrices that govern the quadratic and affine structure of the infinite-horizon problem and provides the natural starting point for the subsequent homogeneous and affine analyses. Expanding the stage cost \eqref{eq:J} and ignoring the constant term independent of the optimizer yields
\begin{equation}
\ell(x,u)=x^\top Qx+2x^\top Nu+u^\top Ru+2c^\top x,
\label{eq:stage_cost}
\end{equation}
with
\begin{align*}
Q&=W\sbs{D}+W\sbs{P}+W\sbs{C}\succeq 0,\\
N&=-\tfrac12\,W\sbs{EN}-W\sbs{C},\\
R&=W\sbs{EX}+\alpha\sbs{F}L_u+W\sbs{C}\succ 0,\\
c&=-W\sbs{D}x_{\mathrm{eq}}.
\end{align*}
Introducing the shifted input
\begin{equation*}
v:=u+R^{-1}N x,
\end{equation*}
the optimal control problem \eqref{eq:prob_full} is equivalent to
\begin{align}
\min_{v}\quad & \widetilde J=\int_{0}^{\infty}\widetilde\ell(x(t),v(t))\,\dif t, \nonumber\\
\mathrm{s.t.}\quad & \dot x=\widetilde A x+v+d,\qquad x(0)=x_0,
\label{eq:prob_full_alt}
\end{align}
where
\begin{equation}
\widetilde A:=A_c-R^{-1}N,
\qquad
\widetilde Q:=Q-NR^{-1}N,
\label{eq:subs}
\end{equation}
and
\begin{equation*}
\widetilde\ell(x,v)=x^\top \widetilde Qx+v^\top Rv+2c^\top x.
\end{equation*}

\section{Solutions of the LQ Optimization Problem: The Homogeneous Case}
\label{sec:sols_homo}

To establish the baseline geometric properties of the recommendation system, we first address the homogeneous free-endpoint problem where the uncontrolled opinion drift and linear cost terms are zero, that is, $d=0\cdot \mathbbm{1}_{nm}$ and $c=0\cdot \mathbbm{1}_{nm}$. The effective quadratic stage cost becomes
\begin{equation}
\widetilde\ell_{sq}(x,v):=x^\top \widetilde Qx+v^\top Rv.
\label{eq:stage_cost_homo}
\end{equation}
The sign structure of $\widetilde\ell_{sq}$ is a first key indicator of the problem regime. In particular, once strict positive definiteness is lost, the classical LQR guarantees no longer apply automatically, and the geometry of the dynamics becomes essential. We can establish simple spectral bounds for when this form is positive definite.

\begin{lemma}
\label{lem:PDP}
Consider $\widetilde\ell_{sq}$ in \eqref{eq:stage_cost_homo}. If
\begin{equation}
\lambda_{m,\mathrm{D}}+\lambda_{m,\mathrm{P}}+\lambda_{m,\mathrm{C}}
>
\frac{(\lambda_{M,\mathrm{EN}}+2\lambda_{M,\mathrm{C}})^2}{4(\lambda_{m,\mathrm{EX}}+\lambda_{m,\mathrm{C}})},
\label{eq:cond_L1}
\end{equation}
then the block matrix $\begin{bmatrix}\widetilde Q& O_{nm} \\ O_{nm} &R\end{bmatrix}$and thus $\widetilde\ell_{sq}(x,v)$ in \eqref{eq:stage_cost_homo} are positive definite.
\end{lemma}
\begin{proof}
Given $R \succ 0$,
\begin{equation}
\label{eq:proof_L1_1}
\begin{bmatrix}\widetilde Q&O_{nm}\\O_{nm}&R\end{bmatrix}\succ 0
\iff \widetilde Q\succ 0 \ \text{and}\ R \succ 0.
\end{equation}
Since $L_u\succeq 0$ and $W\sbs{C}\succ 0$,
\begin{align*}
\widetilde Q
&=
(W\sbs{D}+W\sbs{P}+W\sbs{C})
\\
&-
\Big(\tfrac12 W\sbs{EN}+W\sbs{C}\Big)
(\alpha\sbs{F}L_u+W\sbs{EX}+W\sbs{C})^{-1}
\Big(\tfrac12 W\sbs{EN}+W\sbs{C}\Big)
\\
&\succ
\left(
\lambda_{m,\mathrm{D}}+\lambda_{m,\mathrm{P}}+\lambda_{m,\mathrm{C}}
-
\frac{(\lambda_{M,\mathrm{EN}}+2\lambda_{M,\mathrm{C}})^2}{4(\lambda_{m,\mathrm{EX}}+\lambda_{m,\mathrm{C}})}
\right)I_{nm},
\end{align*}
and the last matrix is positive definite whenever \eqref{eq:cond_L1} holds.
\end{proof}

\begin{corollary}
\label{cor:PDP2}
Suppose $Q,R$ and $N$
 are simultaneously orthogonally diagonalizable, that is,
there exist an orthogonal $U\in \reals^{nm \times nm}$ and $q,r,s\in \reals^{nm}_{>0}$ such that
$Q=U\mathrm{diag}(q)U^\top$, $R=U\mathrm{diag}(r)U^\top$,
$N=U\mathrm{diag}(-s)U^\top$.
Then
\begin{equation*}
\begin{bmatrix}\widetilde Q & O_{nm}\\O_{nm} & R\end{bmatrix}\succ0
\iff
q_i \;>\; \frac{s_i^2}{\,r_i}\ \ \forall i\in [nm].
\end{equation*}
\end{corollary}
\begin{proof}
By \eqref{eq:proof_L1_1},
\begin{equation*}
\begin{bmatrix}\widetilde Q & O_{nm}\\O_{nm} & R\end{bmatrix}\succ0
\iff
\widetilde Q\succ0 \ \text{and}\ R\succ0.
\end{equation*}
Since $Q$, $R$, and $N$ are simultaneously orthogonally diagonalizable,
\begin{equation*}
\begin{aligned}
\widetilde Q&=
Q-NR^{-1}N \\&=
U\left(\diag(q)-\diag(s)\diag(r)^{-1}\diag(s)\right)U^\top.
\end{aligned}
\end{equation*}
Hence $\widetilde Q\succ0$ is equivalent to $q_i>s_i^2/r_i$ for every $i\in[nm]$.
\end{proof}

\begin{remark}
If the weights are homogeneous for each topic, i.e., $Q = Q_t \otimes I_n, W_{\mathrm{EN}} = W_t \otimes I_n, R = R_t \otimes I_n + \alpha_{\mathrm{F}}\, I_m \otimes L_{\mathrm{sym}}$,
then $Q$, $W_{\mathrm{EN}}$, and $R$ pairwise commute since it holds that, given $A\in\reals^{m\times m}, \, B\in\reals^{n\times n}$, $(A\otimes I_n)(I_m\otimes B) = A\otimes B = (I_m\otimes B)(A\otimes I_n)$,
which in turn it implies that they are simultaneously orthogonally diagonalizable. Matrices $Q$, $R$, and $N$ are also simultaneously orthogonally diagonalizable when $J_F=0$
\end{remark}

\begin{remark}
\label{rem:semi}
Lemma~\ref{lem:PDP} and Corollary~\ref{cor:PDP2} can be restated with nonstrict inequalities in the bounds to obtain conditions for positive semidefiniteness of $\widetilde\ell_{sq}$. Indeed, $\diag(\widetilde Q,R)\succeq 0$ if, and only if $\widetilde Q\succeq 0$ and $R\succ 0$, with $R \succ 0$ being a standing assumption. Thus, replacing the strict bounds in Lemma~\ref{lem:PDP} and Corollary~\ref{cor:PDP2} by the corresponding nonstrict versions yields sufficient, and necessary and sufficient conditions, respectively, for $\widetilde\ell_{sq}$ to be nonnegative definite.
\end{remark}

Since the input matrix in the reduced dynamics \eqref{eq:prob_full_alt} is $B=I_{nm}$, the pair $(\widetilde A,B)$ is controllable. Thus, when the reduced quadratic stage cost $\widetilde\ell_{sq}$ in \eqref{eq:stage_cost_homo} is strictly positive definite, the homogeneous problem belongs to the classical continuous-time infinite-horizon LQ setting. In the present paper, however, we focus on what happens when this strict positivity is lost. In that case, the sign structure of $\widetilde\ell_{sq}$ is no longer sufficient, by itself, to determine the qualitative behavior of the optimal closed loop: the geometry of the dynamics and the directions left unpenalized by the cost become essential.

This loss of strict positivity leads to two different nonclassical situations. If $\widetilde Q\succeq0$ is singular, then some state directions are not directly penalized, and detectability of $(\widetilde Q^{1/2},\widetilde A)$ becomes the relevant condition ruling out unstable invisible modes. If, instead, $\widetilde Q$ is sign-indefinite, then the problem leaves the standard stabilizing LQR framework and must be treated through the free-endpoint indefinite LQ theory recalled below. In that regime, one must distinguish between finiteness of the infimum, attainability of an optimal input, and stability of the closed-loop matrix induced by the corresponding free-endpoint solution. Hence, Lemma~\ref{lem:PDP} and Corollary~\ref{cor:PDP2} should be read as design guardrails on the weights: when these guardrails are violated, optimality of the prescribed cost no longer automatically implies bounded or convergent opinion trajectories.

If instead $\widetilde Q$ has negative eigenvalues, then $\widetilde\ell_{sq}$ is sign-indefinite, and the classical convex stabilizing picture breaks down more radically. To treat the unconstrained infinite-horizon homogeneous problem in this regime, we exploit the geometric free-endpoint results of \cite{trentelman1989regular}. Throughout this subsection, admissible controls are understood in the free-endpoint sense of \cite{trentelman1989regular}. We examine the symmetric extremal solutions of the algebraic Riccati equation (ARE)
\begin{equation}
\label{eq:ARE-min}
\widetilde A^\top P + P\,\widetilde A - P\,R^{-1}P + \widetilde Q = 0.
\end{equation}
Let $P_-$ denote the minimal antistabilizing solution, and $P_+$ denote the maximal stabilizing solution. Define
\begin{equation*}
A_-:=\widetilde A - R^{-1}P_-,
\end{equation*}
and the subspace
\begin{equation}
\label{eq:trentelman_subspace}
\mathcal N:=(\ker P_- \mid A_-)\cap \mathcal{X}^+(A_-),
\end{equation}
where $(\ker P_- \mid A_-)$ is the largest $A_-$-invariant subspace contained in $\ker P_-$, and $\mathcal{X}^+(A_-)$ is the $A_-$-invariant subspace spanned by generalized eigenvectors with $\Re(\lambda)\ge0$. Under the regularity hypotheses that ensure existence of the extremal solutions $P_-$ and $P_+$ and applicability of the free-endpoint construction in \cite{trentelman1989regular}, the homogeneous free-endpoint value is generated by the distinguished symmetric solution $P_\star$ of \eqref{eq:ARE-min} supported by $\mathcal N$. Concretely, letting
\begin{equation*}
\green{
\Delta:=P_+-P_-,
}
\end{equation*}
\green{and denoting by $\Pi_{\mathcal N}$ the projector onto $\mathcal N$ along}
\begin{equation*}
\green{
\Delta^{-1}(\mathcal N^\perp)
:=
\{x\in\reals^{nm}:\ \Delta x\in\mathcal N^\perp\},
}
\end{equation*}
\green{the supported solution is the symmetric ARE solution characterized by}
\begin{equation}
\label{eq:trent}
\green{
P_\star
=
P_-\,\Pi_{\mathcal N}
+
P_+\,(I-\Pi_{\mathcal N}).
}
\end{equation}
\green{In particular, $P_\star$ is not introduced here as a standard stabilizing LQR solution, but as the specific free-endpoint solution singled out by the geometric construction in \cite{trentelman1989regular}.} The homogeneous free-endpoint value is then
\begin{equation*}
V_{\mathrm{h}}(x_0)=x_0^\top P_\star x_0.
\end{equation*}
Moreover, optimal controls exist for all initial conditions if and only if
\begin{equation*}
\ker (P_+-P_-) \subseteq \ker P_-,
\end{equation*}
and, whenever they exist, they are generated by the static feedback law
\begin{equation*}
v^\star(x)=-\,R^{-1}P_\star x.
\end{equation*}
If the above kernel inclusion fails, then the homogeneous problem may still have a finite infimum, but that infimum is not attained by any admissible input.

\section{Solutions of the LQ Optimization Problem: The Affine Case}
\label{sec:sols_affine}

We now return to the complete dynamics \eqref{eq:prob_full_alt}, reintroducing the intrinsic network drift ($d \neq 0$) and the linear cost objectives ($c \neq 0$). \green{Because of the affine terms $d$ and $c$, the raw integral $\int_0^\infty \widetilde\ell(x(t),v(t))\,\dif t$ need not be finite even when a stationary affine correction is meaningful. Keeping the free-endpoint philosophy of Section~\ref{sec:sols_homo}, we use the stationary Hamiltonian identity only as a direct algebraic verification device. Fixing a symmetric solution $P_\star$ of the homogeneous Riccati equation \eqref{eq:ARE-min}, we look for a constant $\rho\in\reals$ and a quadratic-affine function}
\begin{equation}
\label{eq:ans_aff}
\green{V(x)=x^\top P_\star x+2p^\top x}
\end{equation}
\green{such that}
\begin{equation}
\label{eq:HJB}
\min_v \left[ \widetilde{\ell}(x,v) + \nabla V(x)^\top (\widetilde{A}x + v + d) \right] = \rho
\end{equation}
\green{holds pointwise in $x$.}

\green{Equation \eqref{eq:HJB} is consistent with the stationary Hamilton--Jacobi--Bellman formalism used in standard dynamic-programming treatments. In the present paper, however, it is not invoked as the conclusion of a general viscosity-solution existence theorem on $\reals^{nm}$ as in \cite{bardi1997optimal}. Instead, it is used only as a pointwise stationary Hamiltonian identity to be verified within a quadratic-affine ansatz.}

\green{Fix a symmetric ARE solution $P_\star$ of \eqref{eq:ARE-min}. In the homogeneous indefinite free-endpoint problem, $P_\star$ can be chosen as in \cite{trentelman1989regular}, as the distinguished supported solution introduced in Section~\ref{sec:sols_homo}. The quadratic-affine ansatz \eqref{eq:ans_aff} is then a problem-specific specialization of the stationary Hamiltonian identity to the present affine-quadratic setting.}

\green{The next result isolates the correct algebraic consistency condition of the affine quadratic-affine stationary synthesis.}

{\color{black}\begin{proposition}
\label{prop:affine_consistency}
Assume $P_\star=P_\star^\top$ solves \eqref{eq:ARE-min}, and define
\begin{equation*}
A_{\mathrm{cl}}:=\widetilde A-R^{-1}P_\star,
\qquad
g:=P_\star d+c.
\end{equation*}
Then the stationary identity \eqref{eq:HJB} admits a quadratic-affine solution $V(x)=x^\top P_\star x+2p^\top x$ if and only if
\begin{equation}
\label{eq:compat}
g\in \operatorname{im}(A_{\mathrm{cl}}^\top).
\end{equation}
In that case, $p$ solves
\begin{equation}
\label{eq:bias-eq}
A_{\mathrm{cl}}^\top p+g=0,
\end{equation}
the associated stationary feedback is
\begin{equation}
\label{eq:v-star-affine}
v^\star(x)=-R^{-1}(P_\star x+p),
\end{equation}
and the constant in \eqref{eq:HJB} is
\begin{equation*}
\rho=2p^\top d-p^\top R^{-1}p.
\end{equation*}
\end{proposition}

\begin{proof}
Substituting the ansatz
\begin{equation*}
V(x)=x^\top P_\star x+2p^\top x
\end{equation*}
into \eqref{eq:HJB} gives
\begin{equation*}
\nabla V(x)=2P_\star x+2p.
\end{equation*}
Since $R\succ0$, the minimization with respect to $v$ is well posed and yields
\begin{equation*}
v^\star(x)=-R^{-1}(P_\star x+p),
\end{equation*}
which is \eqref{eq:v-star-affine}. Replacing $v$ by $v^\star(x)$ in \eqref{eq:HJB}, the quadratic terms cancel by construction of $P_\star$ through the Riccati equation \eqref{eq:ARE-min}. The remaining expression is
\begin{equation*}
2x^\top(A_{\mathrm{cl}}^\top p+g)+2p^\top d-p^\top R^{-1}p.
\end{equation*}
Hence the stationary identity holds for all $x$ if and only if
\begin{equation*}
A_{\mathrm{cl}}^\top p+g=0.
\end{equation*}
This linear equation is solvable if and only if $g\in\operatorname{im}(A_{\mathrm{cl}}^\top)$, equivalently $z^\top g=0$ for every $z\in\ker(A_{\mathrm{cl}})$. When it is solvable, the constant term is absorbed into $\rho$, namely
\begin{equation*}
\rho=2p^\top d-p^\top R^{-1}p.
\end{equation*}
\end{proof}}

\begin{remark}
\label{rem:boundary}
\green{No Hurwitz assumption on $A_{\mathrm{cl}}$ is imposed in Proposition~\ref{prop:affine_consistency}. If $V$ and $v^\star$ satisfy \eqref{eq:HJB}, then along any closed-loop trajectory of}
\begin{equation*}
\green{\dot x=A_{\mathrm{cl}}x+d-R^{-1}p}
\end{equation*}
\green{one has}
\begin{equation*}
\widetilde\ell(x(t),v^\star(x(t)))-\rho
=
-\frac{\dif}{\dif t}V(x(t)).
\end{equation*}
\green{Hence, for every $T>0$,}
\begin{equation*}
\int_0^T \big(\widetilde\ell(x(t),v^\star(x(t)))-\rho\big)\,\dif t
=
V(x(0))-V(x(T)).
\end{equation*}
\green{Therefore $\rho$ may be interpreted as an average cost rate along trajectories for which $V(x(T))/T\to0$ \citep{bardi1997optimal}, but such an asymptotic interpretation is not imposed a priori in the present free-endpoint setting.}
\end{remark}

The proposition shows that the affine stationary synthesis is governed by the bias equation \eqref{eq:bias-eq}. When \eqref{eq:compat} holds, the affine forcing $g=P_\star d+c$ can be absorbed by a constant bias vector $p$, and the stationary quadratic-affine construction is well defined. When \eqref{eq:compat} fails, the forcing has a component along a null direction of the closed-loop operator $A_{\mathrm{cl}}$, and no quadratic-affine stationary identity of the above form exists.

\section{Pathological Regimes and Free-Endpoint Behaviors}
\label{sec:problems}

The aim of this section is to discuss, in a general way, the different pathologies that can arise when the weights in the performance index are chosen outside the strictly positive-definite regime identified by Lemma~\ref{lem:PDP} and Corollary~\ref{cor:PDP2}. The detailed algebraic constructions are collected in Appendix~\ref{app:examples}, while here, we focus on the qualitative meaning of the examples and on the consequences of entering an indefinite or semidefinite LQ regime.

The common mechanism behind all the examples is the following. Once the effective reduced stage cost loses strict positive definiteness, minimizing the cost no longer automatically enforces the qualitative properties that motivated the recommendation design. The problem may still be meaningful from the viewpoint of Riccati theory, and it may admit extremal Riccati solutions, finite values, and in some cases even an attained free-endpoint optimum. However, these objects need not correspond to a recommendation policy that stabilizes the opinion dynamics, moderates disagreement, or remains implementable as an admissible input. Thus, outside the strictly positive-definite regime, there can be a mismatch between the mathematical optimization criterion and the intended behavioral objective of the platform.

The first type of pathology appears in the sign-indefinite regime. In this case, some state directions are not sufficiently penalized by the reduced cost and may even be effectively rewarded. The free-endpoint indefinite LQ theory of~\cite{trentelman1989regular} may still provide a distinguished Riccati solution and, in favorable cases, an optimal feedback. Nevertheless, the closed-loop matrix generated by that feedback can retain an unstable free mode. This means that the optimization problem is well posed in the free-endpoint sense and the optimizer exists, but the resulting recommendation policy does not stabilize the opinion dynamics. From the viewpoint of the application, this is a central warning: in an indefinite regime, cost optimality alone does not guarantee boundedness of opinions, moderation, or suppression of destabilizing amplification mechanisms. This phenomenon is illustrated in Appendix~\ref{app:examples}, Example~\ref{ex:Ex1}.

A second and distinct pathology is loss of attainability. In this case, the reduced stage cost is again sign-indefinite, and the algebraic Riccati equation still admits the extremal solutions $P_-$ and $P_+$. Hence the problem is not ill posed at the level of Riccati solvability. However, the attainability condition
\begin{equation*}
\ker(P_+-P_-)\subseteq\ker(P_-)
\end{equation*}
fails. As a consequence, the homogeneous problem has a finite infimum, but no admissible control attains it for initial conditions with a component along the nonattainable direction. The optimization problem therefore defines an ideal lower bound, but this bound cannot be realized by any admissible recommendation input. This is not merely an instability issue, and the optimal recommendation policy itself does not actually exist as an implementable object. The corresponding construction is given in Appendix~\ref{app:examples}, Example~\ref{ex:finite_unattained}.

The third pathology shows that even the semidefinite case can be problematic. Here the reduced stage cost is nonnegative, so the issue is not the presence of negative directions in the cost. Instead, the problem is that some state directions are invisible to the performance index. If $\widetilde Q\succeq0$ but $\widetilde Q\not\succ0$, and the pair $(\widetilde Q^{1/2},\widetilde A)$ is not detectable, then an unstable mode can remain completely unpenalized. In the extreme case, the zero input is optimal because it minimizes control effort, even though the corresponding state trajectory diverges. This example shows that semidefiniteness is not automatically benign: without detectability, an optimal recommendation policy may simply ignore the unstable opinion directions that the design was meant to control. This case is presented in Appendix~\ref{app:examples}, Example~\ref{ex:semidefinite_nondetectable}.

The affine terms introduce another layer of possible failures. The fourth example shows that affine compatibility can hold exactly, while the free-endpoint instability survives. In that case, the bias equation is solvable, the stationary quadratic-affine identity is well defined, and the affine feedback can be constructed. However, this algebraic consistency does not imply that the resulting closed-loop dynamics are stable. If the homogeneous free-endpoint feedback leaves an unstable free mode active, an affine translation can preserve the same unstable direction. Therefore, affine compatibility alone is not enough to recover the qualitative properties that the recommendation design was supposed to enforce. This mechanism is shown in Appendix~\ref{app:examples}, Example~\ref{ex:free_unstable_aff}.

The fifth example isolates a different obstruction, which is specific to the affine stationary synthesis. After fixing a homogeneous Riccati solution $P_\star$, define
\begin{equation*}
A_{\mathrm{cl}}:=\widetilde A-R^{-1}P_\star,
\qquad
g:=P_\star d+c.
\end{equation*}
The quadratic-affine stationary construction requires the bias equation
\begin{equation*}
A_{\mathrm{cl}}^\top p+g=0
\end{equation*}
to be solvable, or equivalently
\begin{equation*}
g\in\operatorname{im}(A_{\mathrm{cl}}^\top).
\end{equation*}
If this condition fails, no vector $p$ can absorb the affine forcing, and no stationary quadratic-affine identity of the intended form can be constructed. Here the failure occurs before any question of closed-loop stability: the stationary affine controller itself cannot be obtained from the proposed quadratic-affine ansatz. This obstruction is constructed explicitly in Appendix~\ref{app:examples}, Example~\ref{ex:affine_no_control}.

Taken together, the five examples show that leaving the strictly positive-definite regime can produce several different forms of failure. One may obtain (i) an optimal policy that exists but leaves unstable modes active, (ii) a finite value that is not attained by any admissible policy, (iii) a semidefinite optimum that ignores unstable opinion directions, (iv) an affine correction that is algebraically compatible but preserves the unstable homogeneous mode, or (v) an affine problem for which the stationary correction is algebraically inconsistent. These are distinct control-theoretic phenomena, but they have the same interpretation for recommendation design: the selected performance index no longer reliably encodes the intended behavior of the platform.

For this reason, Lemma~\ref{lem:PDP} and Corollary~\ref{cor:PDP2} should not be read only as technical sufficient conditions. In the present modeling framework, they act as design guardrails on the relative weights assigned to engagement, polarization, baseline preservation, recommendation mismatch, graph-based exposure regularization, and control effort. If engagement is over-rewarded relative to the regularizing terms, the resulting infinite-horizon problem may still possess a rich Riccati structure, but the optimal synthesis can lose its intended technological meaning: it may fail to produce bounded opinion trajectories, meaningful regulation of disagreement, or an implementable recommendation policy.

\section{Discussion and conclusion}
\label{sec:disc}
In this paper, we formulated recommendation design over social networks as a closed-loop optimal control problem with explicit trade-offs between engagement, polarization, baseline preservation, graph-based exposure regularization, control effort, and recommendation mismatch. On the mathematical side, the homogeneous reduced problem is connected to the established free-endpoint indefinite LQ framework, while the affine extension is treated more conservatively through a direct quadratic-affine stationary verification calculation. This separation allows us to distinguish between pathologies that are already present in the homogeneous problem and obstructions that arise only when the affine drift and the linear terms in the cost are reintroduced.

Our analysis and examples promote the idea that the spectral conditions given in Lemma~\ref{lem:PDP} and Corollary~\ref{cor:PDP2} should be interpreted as design constraints on the performance-index weights. These conditions tie the engagement reward to the amount of penalization required for polarization, deviation from the uncontrolled baseline, recommendation mismatch, graph-based exposure regularization, and control effort. If engagement is overweighted relative to these regularizing terms, the effective quadratic form may become sign-indefinite or semidefinite, pushing the design into a free-endpoint LQ regime in which closed-loop stability is no longer guaranteed by the performance index itself. In that regime, stability must be enforced as an additional constraint, which is less satisfactory than encoding the desired user--platform interaction directly through a well-posed cost.

For the affine formulation, the stationary Hamiltonian identity is used only as an algebraic verification tool built on a fixed homogeneous Riccati solution. This choice preserves the free-endpoint viewpoint of~\cite{trentelman1989regular}, and the affine correction may be compatible or incompatible, and even when it is compatible, it may still preserve a free unstable mode, as illustrated in Example~\ref{ex:free_unstable_aff}. Thus, affine consistency should not be confused with recovery of the qualitative properties that motivated the recommendation design.

The complementary positive-definite regime remains an important direction of ongoing work. When $\widetilde\ell_{sq}(x,v)$ is strictly positive definite, the reduced homogeneous problem belongs to the classical continuous-time infinite-horizon LQ framework~\citep{anderson2007optimal}. The cost is coercive, bounded below, and the usual stabilizing interpretation of the optimal feedback is recovered. Similarly, when the spectral bounds of Lemma~\ref{lem:PDP} and Corollary~\ref{cor:PDP2} are satisfied, the reduced affine-quadratic problem lies in the standard strictly convex regime. In this regime, the same spectral conditions can be used not only as diagnostic tests, but also as constructive constraints for selecting admissible weights in the performance index.

A natural next step is therefore to exploit this well-posed region to synthesize structured recommendation policies. In particular, one would like to impose locality, sparsity, or graph-induced information constraints on the feedback law, so that the recommendation policy is compatible with the social-network structure and with distributed or partially local platform architectures. The development of structured and decentralized certificates for recommendation synthesis over social networks is currently under investigation. Such certificates would aim to preserve the stability and well-posedness guarantees of the positive-definite regime while avoiding centralized feedback laws that may be unrealistic for large-scale platforms.

The present work also has modeling limitations that point to further extensions. We assumed full-state feedback and a time-invariant, known interaction structure, encoded in $L$, $C$, and $A_a$. In practice, the social graph, the inter-topic logic, and the anchoring strengths should be estimated from data, and may be time-varying or uncertain. Extending the analysis to partial observability, output feedback, model uncertainty, and recursively learned parameters is therefore a natural next step.

Likewise, we considered quadratic penalties, static state feedback, and unconstrained inputs. Incorporating additional constraints on exposure, fairness, safety, or variety of the recommended content would better reflect real-world platforms, but would also lead to constrained, nonlinear, or dynamic optimal control problems. Similarly, constraints on the pool of admissible recommendation inputs, such as saturation, quantization, or restrictions on the input matrix, would be necessary to model platforms in which the recommender cannot propose arbitrary opinion-shaping inputs.

Finally, while the weights in the performance index can be heterogeneous across users and topics, they are treated here as fixed design and modeling parameters. Embedding the proposed spectral guardrails into learning-based recommendation pipelines would allow engagement-oriented objectives to be tuned subject to stability and polarization constraints. This would connect the present model-based approach with data-driven recommendation architectures~\citep{chandrasekaran2024network}, and would provide a route toward recommendation systems whose engagement objectives are optimized without leaving the well-posed region identified by the control-theoretic analysis.

\bibliography{Bib}

\appendix

\section{Examples}
\label{app:examples}
\begin{example}[Attained optimum with unstable free mode]
\label{ex:Ex1}
Consider the optimization problem \eqref{eq:prob_full_alt}. Let $n=1$, $m=2$, and choose
\begin{align*}
&A_a=1,\qquad
C=\begin{bmatrix}
1 & \xi\\
0 & \tfrac12
\end{bmatrix},
\qquad
d=c=\begin{bmatrix}0\\0\end{bmatrix},
\qquad
R=I_2,\\
&N=\diag(-2-\eta,\,-\tfrac52+\eta),\\
&Q=\diag\!\Big((2+\eta)^2,\;(-\tfrac52+\eta)^2-\beta\Big),
\end{align*}
with $0<\eta<\frac52$,
$\xi\in(-1,0)\cup(0,1)$, and
$0<\beta<\min\!\left\{\eta^2,\big(\tfrac52-\eta\big)^2\right\}$.
Then, by \eqref{eq:subs},
\begin{equation*}
\widetilde A=
\begin{bmatrix}
\eta & \xi\\[2pt]
0 & -\eta
\end{bmatrix},
\qquad
\widetilde Q=\diag(0,-\beta).
\end{equation*}
Hence the reduced homogeneous stage cost is sign-indefinite and the first state is completely unpenalized.

Let
\begin{equation*}
P=P^\top=
\begin{bmatrix}
P_{11}&P_{12}\\[2pt]
P_{12}&P_{22}
\end{bmatrix}
\end{equation*}
solve the ARE \eqref{eq:ARE-min}. Expanding \eqref{eq:ARE-min} gives
\begin{equation}
\label{ex1:ARE-system-fixed}
\begin{cases}
2\eta\,P_{11}-(P_{11}^2+P_{12}^2)=0,\\[2pt]
\xi\,P_{11}-P_{12}(P_{11}+P_{22})=0,\\[2pt]
2\xi\,P_{12}-2\eta\,P_{22}-(P_{12}^2+P_{22}^2)-\beta=0.
\end{cases}
\end{equation}

First set $P_{12}=0$. Since $\xi\neq 0$, the second equation implies $P_{11}=0$, and the third becomes
\begin{equation*}
P_{22}^2+2\eta P_{22}+\beta=0.
\end{equation*}
Defining
\begin{equation*}
\Delta:=\sqrt{\eta^2-\beta}\in(0,\eta),
\end{equation*}
the two diagonal ARE solutions are
\begin{equation*}
\diag(0,-\eta-\Delta),
\qquad
\diag(0,-\eta+\Delta).
\end{equation*}
Among them,
\begin{equation*}
P_-:=\diag(0,-\eta-\Delta)
\end{equation*}
is antistabilizing, because
\begin{equation*}
A_-:=\widetilde A-P_-=
\begin{bmatrix}
\eta & \xi\\[2pt]
0 & \Delta
\end{bmatrix}
\end{equation*}
has eigenvalues $\eta>0$ and $\Delta>0$.

Now consider $P_{12}\neq 0$. Parameterize the first equation in \eqref{ex1:ARE-system-fixed} as
\begin{equation}
\label{ex1:param-fixed}
P_{11}=\frac{2\eta}{1+\iota^2},
\qquad
P_{12}=\frac{2\eta\,\iota}{1+\iota^2},
\qquad
\iota\in\reals\setminus\{0\}.
\end{equation}
From the second equation,
\begin{equation}
\label{ex1:p22-fixed}
P_{22}
=
\frac{\xi P_{11}}{P_{12}}-P_{11}
=
\frac{\xi}{\iota}-\frac{2\eta}{1+\iota^2}.
\end{equation}
Substituting \eqref{ex1:param-fixed}--\eqref{ex1:p22-fixed} into the third equation yields
\begin{equation*}
\beta\,\iota^2-2\xi\eta\,\iota+\xi^2=0.
\end{equation*}
Hence, for $\gamma\in\{+1,-1\}$,
\begin{equation*}
\iota_\gamma=\frac{\xi(\eta+\gamma\Delta)}{\beta}.
\end{equation*}
Define
\begin{equation*}
\zeta_\gamma:=\beta^2+\xi^2(\eta+\gamma\Delta)^2.
\end{equation*}
The corresponding two off-diagonal ARE solutions are
\begin{equation*}
P^{(\gamma)}=
\begin{bmatrix}
\dfrac{2\eta\beta^2}{\zeta_\gamma}
&
\dfrac{2\eta\beta\xi(\eta+\gamma\Delta)}{\zeta_\gamma}
\\[10pt]
\dfrac{2\eta\beta\xi(\eta+\gamma\Delta)}{\zeta_\gamma}
&
\eta-\gamma\Delta-\dfrac{2\eta\beta^2}{\zeta_\gamma}
\end{bmatrix}.
\end{equation*}

A direct calculation gives $\operatorname{tr}(\widetilde A-P^{(\gamma)})=-(\eta-\gamma\Delta)$ and
$\det(\widetilde A-P^{(\gamma)})=-\gamma\,\eta\Delta.$
Therefore:
$\widetilde A-P^{(-)}$ is Hurwitz,
and
$\widetilde A-P^{(+)}$ has one positive and one negative eigenvalue.
Hence the maximal stabilizing solution is
\begin{equation*}
\widehat P_+:=P^{(-)}.
\end{equation*}

We now determine the supported free-endpoint solution. Since
\begin{equation*}
\ker(P_-)=\mathrm{span}(e_1),
\qquad
A_-e_1=\eta e_1,
\end{equation*}
it follows that
\begin{equation*}
(\ker P_-\mid A_-)=\mathrm{span}(e_1).
\end{equation*}
Because $\eta,\Delta>0$, one has $\mathcal X^+(A_-)=\reals^2$, and therefore the free subspace is
\begin{equation*}
\mathcal N=(\ker P_-\mid A_-)\cap \mathcal X^+(A_-)=\mathrm{span}(e_1).
\end{equation*}

Write
\begin{equation*}
\widehat P_+=
\begin{bmatrix}
a & b\\
b & c
\end{bmatrix},
\end{equation*}
where
\begin{equation*}
a=\frac{2\eta(\eta+\Delta)^2}{(\eta+\Delta)^2+\xi^2},
\qquad
b=\frac{2\eta\xi(\eta+\Delta)}{(\eta+\Delta)^2+\xi^2},
\qquad
c=\eta+\Delta-a.
\end{equation*}
Let
\begin{equation*}
\widehat\Delta:=\widehat P_+-P_-.
\end{equation*}
Since $\mathcal N^\perp=\mathrm{span}(e_2)$, the correct supporting projector is the projection onto $\mathcal N$ along $\widehat\Delta^{-1}(\mathcal N^\perp)$. Now
\begin{equation*}
\widehat\Delta^{-1}(\mathcal N^\perp)
=
\{x\in\reals^2:\ \widehat\Delta x\in\mathrm{span}(e_2)\}
=
\ker\!\big(e_1^\top\widehat\Delta\big)
=
\ker\!\begin{bmatrix}a&b\end{bmatrix}.
\end{equation*}
Hence
\begin{equation*}
\widehat\Delta^{-1}(\mathcal N^\perp)
=
\mathrm{span}\!\begin{bmatrix}-b\\ a\end{bmatrix},
\end{equation*}
and the associated projector is
\begin{equation*}
\Pi_{\mathcal N}
=
e_1\begin{bmatrix}1&b/a\end{bmatrix}
=
e_1\begin{bmatrix}1&\dfrac{\xi}{\eta+\Delta}\end{bmatrix}.
\end{equation*}

The supported free-endpoint ARE solution is therefore
\begin{equation*}
P_\star
=
P_-\,\Pi_{\mathcal N}
+
\widehat P_+(I-\Pi_{\mathcal N});
\end{equation*}
see \cite{trentelman1989regular} for further detail.
A direct computation yields
\begin{equation*}
P_\star=\diag(0,-\eta+\Delta).
\end{equation*}
Moreover,
\begin{equation*}
\ker(\widehat P_+-P_-)=\{0\}\subseteq \ker(P_-),
\end{equation*}
so, by the free-endpoint result recalled in Section~\ref{sec:sols_homo}, optimal controls exist for every initial condition and are generated by
\begin{equation*}
v^\star(x)=-P_\star x.
\end{equation*}

The implemented closed-loop matrix is
\begin{equation*}
A_{\mathrm{cl}}
=
\widetilde A-P_\star
=
\begin{bmatrix}
\eta & \xi\\[2pt]
0 & -\Delta
\end{bmatrix},
\end{equation*}
hence
\begin{equation*}
\sigma(A_{\mathrm{cl}})=\{\eta,-\Delta\}.
\end{equation*}
Therefore the supported free-endpoint optimal controller exists, but the closed loop still retains the unstable mode $\eta>0$.
\end{example}

\begin{example}[Finite but unattained homogeneous infimum]
\label{ex:finite_unattained}
Consider \eqref{eq:prob_full_alt} with $n=1$, $m=2$, and choose
\begin{equation*}
A_a=2,\quad
C=I_2,\quad
d=c=\begin{bmatrix}0\\0\end{bmatrix},\quad
R=I_2,
\end{equation*}
\begin{equation*}
N=-2I_2,
\qquad
Q=\operatorname{diag}(3,4).
\end{equation*}
Then
\begin{equation*}
\widetilde A=-I_2,
\qquad
\widetilde Q=\operatorname{diag}(-1,0).
\end{equation*}

The reduced dynamics is
\begin{equation*}
\dot x=-x+v,
\end{equation*}
and the cost is
\begin{equation*}
\widetilde J(v)=\int_0^\infty \big(-x_1(t)^2+v_1(t)^2+v_2(t)^2\big)\,dt.
\end{equation*}

The ARE
\begin{equation*}
-\!2P-P^2+\widetilde Q=0
\end{equation*}
has exactly two symmetric solutions,
\begin{equation*}
P_-=\operatorname{diag}(-1,-2),
\qquad
P_+=\operatorname{diag}(-1,0).
\end{equation*}
Hence
\begin{equation*}
\ker(P_+-P_-)=\operatorname{span}(e_1),
\qquad
\ker(P_-)=\{0\},
\end{equation*}
so the attainability condition in \cite[Thm. 5]{trentelman1989regular} fails.

Moreover, the first coordinate admits the exact identity
\begin{equation*}
v_1^2-x_1^2=(v_1-x_1)^2+\frac{d}{dt}(x_1^2),
\qquad
\dot x_1=-x_1+v_1.
\end{equation*}
Therefore, for every $T>0$,
\begin{equation*}
\begin{aligned}
&\int_0^T \big(v_1^2-x_1^2+v_2^2\big)\,dt \\
&=
-x_1(0)^2+x_1(T)^2+\int_0^T \big((v_1-x_1)^2+v_2^2\big)\,dt,
\end{aligned}
\end{equation*}
and hence every admissible control satisfies
\begin{equation*}
\widetilde J(v)\ge -x_1(0)^2.
\end{equation*}

This lower bound is not attained when $x_1(0)\neq 0$. Indeed, equality would require
\begin{equation*}
v_1(t)=x_1(t)\quad\text{a.e.},\qquad v_2(t)=0\quad\text{a.e.},
\qquad x_1(T)\to0,
\end{equation*}
but $v_1=x_1$ implies $\dot x_1=0$, so $x_1(t)\equiv x_1(0)\neq0$, which leads to a contradiction.

On the other hand, for any $\varepsilon>0$, the feedback
\begin{equation*}
v_1=(1-\varepsilon)x_1,\qquad v_2=0
\end{equation*}
yields $x_1(t)=e^{-\varepsilon t}x_1(0)$ and cost
\begin{equation*}
\widetilde J_\varepsilon
=
\left(-1+\frac{\varepsilon}{2}\right)x_1(0)^2
\to -x_1(0)^2
\quad\text{as }\varepsilon\rightarrow0.
\end{equation*}
Hence, the infimum is finite but unattained. The geometric obstruction is shown
in Fig.~\ref{fig:finite_unattained}: the nontrivial subspace
$\ker(P_+-P_-)=\operatorname{span}(e_1)$ is not contained in
$\ker(P_-)=\{0\}$. Therefore, every initial condition with a nonzero component
along $e_1$, namely every $x_0$ such that $x_1(0)\neq0$, gives rise to the
nonattainability phenomenon. These initial conditions include those lying exactly on
the line $\operatorname{span}(e_1)$ (in red in Fig.~\ref{fig:finite_unattained}), where the obstruction is purely along the
nonattainable direction, as well as initial conditions on either side of that
line, provided that their projection onto $e_1$ is nonzero. In all these cases,
the cost can approach the lower bound $-x_1(0)^2$, but no admissible input
attains it.
\begin{figure}[tbp]
    \centering
    \includegraphics[width=0.6\linewidth]{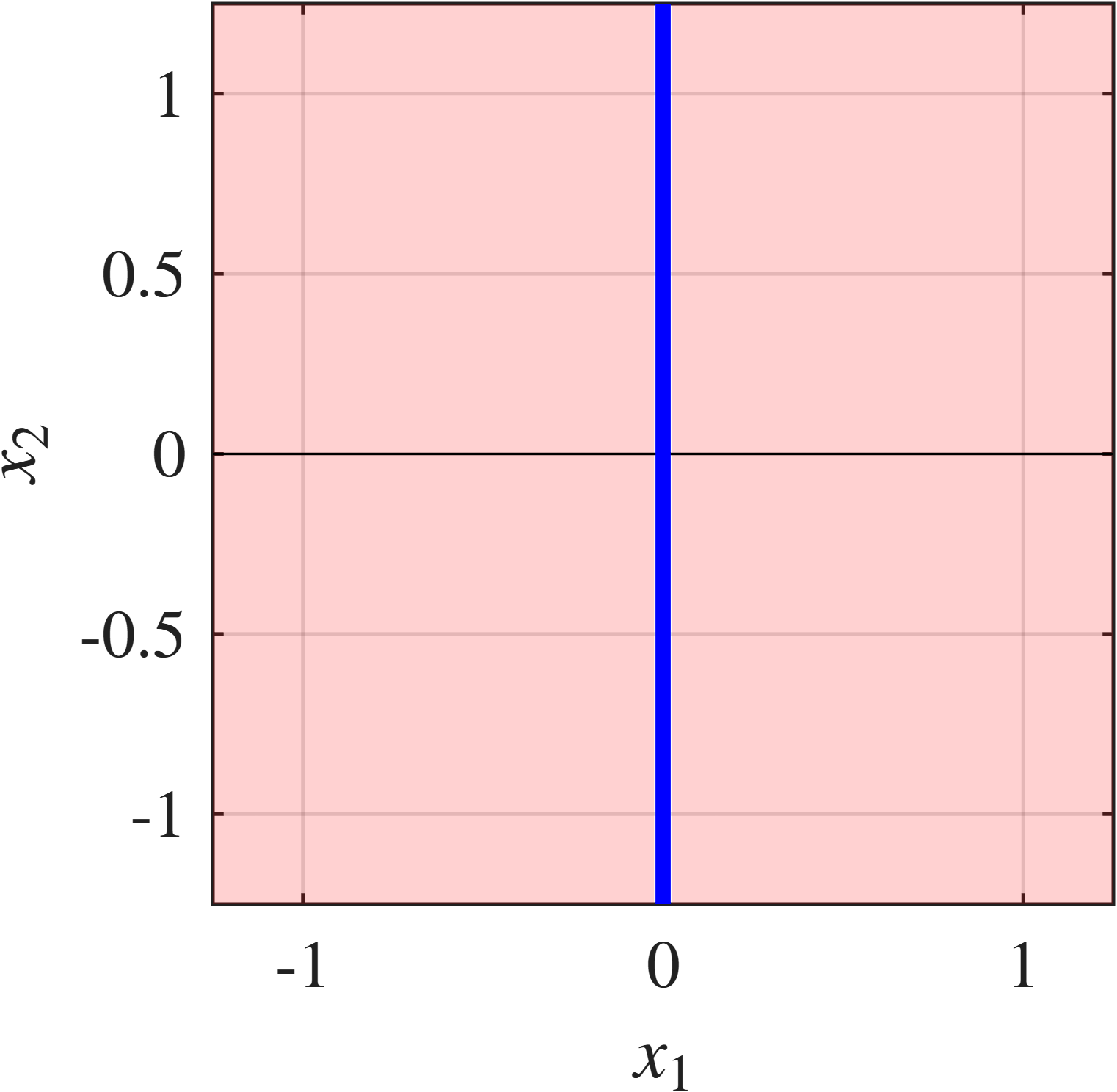}
    \caption{Problematic and non-problematic initial conditions associated with the finite but unattained homogeneous infimum in Example~\ref{ex:finite_unattained}. The light red region represents the set of initial conditions with $x_1(0)\neq0$, for which the infimum is finite but not attained. The blue set represents the subspace $x_1(0)=0$, for which the obstruction is absent and the infimum is attained.}
    \label{fig:finite_unattained}
\end{figure}
\end{example}

\begin{example}[Semidefinite optimum without detectability]
\label{ex:semidefinite_nondetectable}
Consider the optimization problem \eqref{eq:prob_full_alt}, let $n=1$, $m=2$, and choose
\begin{align*}
&A_a= 1, \quad   C=\begin{bmatrix}
1 & \xi \\ \xi & 1
\end{bmatrix}, \\
&L=0, \quad d=c=\begin{bmatrix}0\\0\end{bmatrix},\qquad
R= I_2, \\
&N= -\eta I_2,\quad Q=\eta^2 I_2,
\end{align*}
and thus, given \eqref{eq:subs},
\begin{align*}
    &\widetilde A=\begin{bmatrix}
        -2+\eta & \xi \\[2pt]
        \xi & -2+\eta
    \end{bmatrix}, \quad \widetilde Q= O_2,
\end{align*}
with $\eta>2+ |\xi|>0$, $\xi \in  (-1,0) \cup (0,1)$.
The ARE \eqref{eq:ARE-min} becomes
\begin{equation*}
    \widetilde A P + P \widetilde A - P^2 = 0,\qquad P=P^\top\in\mathbb R^{2\times2}.
\end{equation*}
\green{Since $\widetilde Q=O_2$, the reduced cost is simply
\begin{equation*}
\widetilde J(v)=\int_0^\infty v(t)^\top v(t)\,\dif t \ge 0.
\end{equation*}
Hence the zero input
\begin{equation*}
v^\star(x)\equiv 0
\end{equation*}
is optimal and achieves the value $0$.

Under this optimal policy, the closed-loop matrix is
\begin{equation*}
A_{\mathrm{cl}}=\widetilde A,
\end{equation*}
whose eigenvalues are $\lambda_1>0$ and $\lambda_2>0$. Therefore the optimal free-endpoint policy does not stabilize the system, and the state diverges. Equivalently, the pair $(\widetilde Q^{1/2},\widetilde A)$ is not detectable.}
\end{example}

\begin{example}[Affine-compatible unstable free mode]
\label{ex:free_unstable_aff}
\green{Consider the affine reduced problem \eqref{eq:prob_full_alt}, with $n=10$,
$m=2$, and select the interconnection graph to be complete and undirected with
uniform edge weight $1/5$. Let $L$ be the corresponding graph Laplacian.
Choose
\begin{equation*}
A_a=I_{10},
\qquad
C=
\begin{bmatrix}
1 & \frac{1}{4}\\
0 & \frac{1}{2}
\end{bmatrix}.
\end{equation*}
The matrix $A_c$ in \eqref{eq:model_vec_comp} is
\begin{equation*}
A_c=
\begin{bmatrix}
-L-2I_{10} & \frac14 I_{10}\\
0 & -L-\frac52 I_{10}
\end{bmatrix}.
\end{equation*}

Let
\begin{equation*}
R=I_{20},
\qquad
N=
\operatorname{diag}
\left(
-\frac52 I_{10},
-2I_{10}
\right),
\end{equation*}
and
\begin{equation*}
Q=
\operatorname{diag}
\left(
\frac{25}{4}I_{10},
\frac{63}{16}I_{10}
\right).
\end{equation*}
Hence $Q\succeq0$, $R\succ0$, and $Q,R,N$ are diagonal. Moreover, by
\eqref{eq:subs},
\begin{equation*}
\widetilde A
=
\begin{bmatrix}
\frac12 I_{10}-L & \frac14 I_{10}\\
0 & -\frac12 I_{10}-L
\end{bmatrix}=
\begin{bmatrix}
\widetilde A_{1,1} & \widetilde A_{1,2}\\
0 & \widetilde A_{2,2}
\end{bmatrix},
\end{equation*}
and
\begin{equation*}
\widetilde Q
=
\operatorname{diag}
\left(
0,
-\frac1{16}I_{10}
\right)=
\begin{bmatrix}
0 & 0\\
0 & \widetilde Q_2
\end{bmatrix}.
\end{equation*}
Thus the reduced homogeneous stage cost is sign-indefinite.

Furthermore, let
\begin{equation*}
X_\circ=
\begin{bmatrix}
0.8 & 0.4\\
0.6 & -0.2\\
0.9 & 0.1\\
0.4 & -0.3\\
1.0 & 0.2\\
0.2 & -0.1\\
0.7 & 0.3\\
0.5 & -0.4\\
0.9 & 0.2\\
0.3 & -0.2
\end{bmatrix}.
\end{equation*}

Define
\begin{equation*}
\Pi:=\frac1{10}\mathbbm{1}_{10}\mathbbm{1}_{10}^{\top},
\qquad
\Pi_\perp:=I_{10}-\Pi .
\end{equation*}
Since the graph is complete with uniform edge weight $1/5$,
\begin{equation*}
L=2\Pi_\perp .
\end{equation*}
Consequently,
\begin{equation*}
L\Pi=0,
\qquad
L\Pi_\perp=2\Pi_\perp .
\end{equation*}
Thus $\Pi$ and $\Pi_\perp$ are the spectral projectors of $L$ associated with
the eigenvalues $0$ and $2$, respectively.

Every state and affine forcing can be decomposed along the Laplacian
eigenbasis as
\begin{equation*}
x
=
\sum_{k=0}^{9}
\begin{bmatrix}
\xi_{1,k}u_k\\
\xi_{2,k}u_k
\end{bmatrix},
\qquad
d
=
\sum_{k=0}^{9}
\begin{bmatrix}
d_{1,k}u_k\\
d_{2,k}u_k
\end{bmatrix},
\end{equation*}
with $u_k$, $k \in \{0,\cdots,9\}$, being an eigenvector of $L$. For each fixed $k$, define
\begin{equation*}
\xi_k:=
\begin{bmatrix}
\xi_{1,k}\\
\xi_{2,k}
\end{bmatrix},
\qquad
d_k:=
\begin{bmatrix}
d_{1,k}\\
d_{2,k}
\end{bmatrix}
\end{equation*}
Since $Lu_k=\lambda_k u_k$, each affine modal component evolves as
\begin{equation*}
\dot \xi_k
=
A_{\lambda_k}\xi_k
+
v_k
+
d_k,
\end{equation*}
before the affine feedback is applied, where
\begin{equation*}
A_{\lambda_k}
=
\begin{bmatrix}
\frac12-\lambda_k & \frac14\\
0 & -\frac12-\lambda_k
\end{bmatrix},
\qquad
Q_{\lambda_k}
=
\begin{bmatrix}
0 & 0\\
0 & -\frac1{16}
\end{bmatrix}.
\end{equation*}

For notational compactness, in the following modal Riccati computations we
write $\lambda$ instead of $\lambda_k$. The corresponding modal ARE is
\begin{equation*}
A_\lambda^\top P_\lambda
+
P_\lambda A_\lambda
-
P_\lambda^2
+
Q_\lambda
=0.
\end{equation*}
Writing
\begin{equation*}
P_\lambda
=
\begin{bmatrix}
r_\lambda & s_\lambda\\
s_\lambda & p_\lambda
\end{bmatrix},
\end{equation*}
direct substitution gives
\begin{align*}
2\left(\frac12-\lambda\right)r_\lambda
-r_\lambda^2-s_\lambda^2&=0,\\
\frac14 r_\lambda-2\lambda s_\lambda
-r_\lambda s_\lambda
-s_\lambda p_\lambda&=0,\\
\frac12 s_\lambda
+
2\left(-\frac12-\lambda\right)p_\lambda
-s_\lambda^2-p_\lambda^2-\frac1{16}&=0.
\end{align*}

We first consider the modal problem associated with $\lambda=0$. In this case
\begin{equation*}
A_0=
\begin{bmatrix}
\eta & \xi\\
0 & -\eta
\end{bmatrix},
\qquad
\eta:=\frac12,
\qquad
\xi:=\frac14,
\qquad
\beta:=\frac1{16}.
\end{equation*}
Let
\begin{equation*}
\Delta_0:=\sqrt{\eta^2-\beta}=\frac{\sqrt3}{4}.
\end{equation*}
Since $\beta=\xi^2$, the matrix
\begin{equation*}
P_{-,0}:=
\begin{bmatrix}
0 & 0\\
0 & -\eta-\Delta_0
\end{bmatrix}
\end{equation*}
solves the modal ARE. Indeed,
\begin{equation*}
2(-\eta)(-\eta-\Delta_0)-(-\eta-\Delta_0)^2-\beta
=
\eta^2-\Delta_0^2-\beta
=
0.
\end{equation*}
Moreover,
\begin{equation*}
A_{-,0}:=A_0-P_{-,0}
=
\begin{bmatrix}
\eta & \xi\\
0 & \Delta_0
\end{bmatrix},
\end{equation*}
whose eigenvalues are $\eta>0$ and $\Delta_0>0$. Hence $P_{-,0}$ is the
minimal antistabilizing solution.

The maximal stabilizing solution is
\begin{equation*}
P_{+,0}:=
\begin{bmatrix}
\eta+\Delta_0 & \xi\\
\xi & 0
\end{bmatrix}.
\end{equation*}
Indeed, using $\beta=\xi^2$ and $\Delta_0^2=\eta^2-\beta$, direct
substitution in the modal ARE gives zero, and
\begin{equation*}
A_0-P_{+,0}
=
\begin{bmatrix}
-\Delta_0 & 0\\
-\xi & -\eta
\end{bmatrix},
\end{equation*}
which is Hurwitz.

We now compute the free subspace associated with this modal problem. Since
\begin{equation*}
\ker(P_{-,0})=\operatorname{span}(e_1),
\qquad
A_{-,0}e_1=\eta e_1,
\end{equation*}
one has
\begin{equation*}
(\ker P_{-,0}\mid A_{-,0})
=
\operatorname{span}(e_1).
\end{equation*}
Moreover, since both eigenvalues of $A_{-,0}$ are positive,
\begin{equation*}
\mathcal X^+(A_{-,0})=\mathbb R^2.
\end{equation*}
Therefore
\begin{equation*}
\mathcal N_0
=
(\ker P_{-,0}\mid A_{-,0})
\cap
\mathcal X^+(A_{-,0})
=
\operatorname{span}(e_1).
\end{equation*}
The supporting projection onto $\mathcal N_0$ along
$(P_{+,0}-P_{-,0})^{-1}\mathcal N_0^\perp$ is
\begin{equation*}
\Pi_{\mathcal N_0}
=
\begin{bmatrix}
1 & \dfrac{\xi}{\eta+\Delta_0}\\
0 & 0
\end{bmatrix}.
\end{equation*}
Consequently, the supported free-endpoint solution associated with
$\lambda=0$ is
\begin{equation*}
P_{\star,0}
=
P_{-,0}\Pi_{\mathcal N_0}
+
P_{+,0}(I-\Pi_{\mathcal N_0}).
\end{equation*}
Using
\begin{equation*}
\frac{\xi^2}{\eta+\Delta_0}
=
\frac{\beta}{\eta+\Delta_0}
=
\eta-\Delta_0,
\end{equation*}
one obtains
\begin{equation*}
P_{\star,0}
=
\begin{bmatrix}
0 & 0\\
0 & -\eta+\Delta_0
\end{bmatrix}
=
\begin{bmatrix}
0 & 0\\
0 & -\frac12+\frac{\sqrt3}{4}
\end{bmatrix}.
\end{equation*}

We now consider the modal problems associated with $\lambda=2$. In this case
\begin{equation*}
A_2
=
\begin{bmatrix}
-\alpha & \xi\\
0 & -\delta
\end{bmatrix},
\qquad
\alpha:=\frac32,
\qquad
\delta:=\frac52,
\qquad
\xi:=\frac14.
\end{equation*}
Let
\begin{equation*}
\Delta_2:=\sqrt{\delta^2-\frac1{16}}
=
\frac{3\sqrt{11}}{4}.
\end{equation*}
The stabilizing solution is
\begin{equation*}
P_{+,2}:=
\begin{bmatrix}
0 & 0\\
0 & -\delta+\Delta_2
\end{bmatrix}.
\end{equation*}
Indeed,
\begin{equation*}
2(-\delta)(-\delta+\Delta_2)
-
(-\delta+\Delta_2)^2
-
\frac1{16}
=
\delta^2-\Delta_2^2-\frac1{16}
=
0,
\end{equation*}
and
\begin{equation*}
A_2-P_{+,2}
=
\begin{bmatrix}
-\alpha & \xi\\
0 & -\Delta_2
\end{bmatrix},
\end{equation*}
which is Hurwitz.

The antistabilizing solution can be written explicitly as follows. Define
\begin{equation*}
\theta_-:=\frac{2-\sqrt{11}}{21},
\end{equation*}
and set
\begin{equation*}
P_{-,2}:=
\begin{bmatrix}
r_- & s_-\\
s_- & p_-
\end{bmatrix},
\end{equation*}
where
\begin{equation*}
r_-=-\frac{3}{1+\theta_-^2},
\qquad
s_-=-\frac{3\theta_-}{1+\theta_-^2},
\end{equation*}
and
\begin{equation*}
p_-=-4+\frac{3}{1+\theta_-^2}+\frac{1}{4\theta_-}.
\end{equation*}
Substitution in the modal ARE gives zero, so $P_{-,2}$ solves it.
Moreover,
\begin{equation*}
\sigma(A_2-P_{-,2})
=
\left\{
\frac32,
\frac{3\sqrt{11}}{4}
\right\},
\end{equation*}
and therefore $P_{-,2}$ is antistabilizing. Finally,
\begin{equation*}
\det(P_{-,2})
=
\frac{3(4\theta_- -1)}
{4\theta_-(1+\theta_-^2)}
\neq0.
\end{equation*}
Hence
\begin{equation*}
\ker(P_{-,2})=\{0\}.
\end{equation*}
It follows that
\begin{equation*}
\mathcal N_2
=
(\ker P_{-,2}\mid A_2-P_{-,2})
\cap
\mathcal X^+(A_2-P_{-,2})
=
\{0\}.
\end{equation*}
Therefore the supported free-endpoint solution associated with $\lambda=2$
coincides with the stabilizing branch:
\begin{equation*}
P_{\star,2}=P_{+,2}
=
\begin{bmatrix}
0 & 0\\
0 & -\frac52+\frac{3\sqrt{11}}{4}
\end{bmatrix}.
\end{equation*}

We can now reconstruct the Riccati solution of the networked problem. The
modal solution associated with $\lambda=0$ contributes the scalar
\begin{equation*}
\rho_0:=-\frac12+\frac{\sqrt3}{4},
\end{equation*}
whereas the modal solutions associated with $\lambda=2$ contribute the scalar
\begin{equation*}
\rho_2:=-\frac52+\frac{3\sqrt{11}}{4}.
\end{equation*}
Transforming the modal solution back to the original coordinates gives
\begin{equation*}
P_2
=
\rho_0\Pi+\rho_2\Pi_\perp .
\end{equation*}

The supporting-subspace condition is also inherited from the modal
construction. The only nontrivial modal free subspace is the one associated
with $\lambda=0$, where
\begin{equation*}
\mathcal N_0=\operatorname{span}(e_1),
\end{equation*}
whereas the modal blocks associated with $\lambda=2$ satisfy
\begin{equation*}
\mathcal N_2=\{0\}.
\end{equation*}
Therefore the network-level free subspace is the lift of $\mathcal N_0$
through the eigenvector $u_0$, namely
\begin{equation*}
\mathcal N
=
\operatorname{span}
\left(
\begin{bmatrix}
u_0\\
0
\end{bmatrix}
\right).
\end{equation*}
Thus $P_\star$ is obtained by applying supporting-subspace
construction in \cite{trentelman1989regular} on each modal subsystem and then reconstructing the corresponding
network-level solution.

The closed-loop matrix generated by the homogeneous free-endpoint feedback is
\begin{equation*}
A_{\mathrm{cl}}
=
\widetilde A-P_\star
=
\begin{bmatrix}
\frac12 I_{10}-L & \frac14 I_{10}\\
0 & -\frac12 I_{10}-L-P_2
\end{bmatrix}.
\end{equation*}
Since this matrix is block upper triangular, its spectrum is the union of the
spectra of the two diagonal blocks. The first-topic block satisfies
\begin{equation*}
\frac12 I_{10}-L
=
\frac12\Pi-\frac32\Pi_\perp .
\end{equation*}
Thus the eigenvalue on $\operatorname{im}\Pi$ is $1/2$, while the eigenvalue
on $\operatorname{im}\Pi_\perp$ is $-3/2$.

For the second-topic block,
\begin{equation*}
-\frac12 I_{10}-L-P_2
=
\left(-\frac12-\rho_0\right)\Pi
+
\left(-\frac52-\rho_2\right)\Pi_\perp .
\end{equation*}
By construction,
\begin{equation*}
-\frac12-\rho_0
=
-\frac{\sqrt3}{4}<0,
\end{equation*}
and
\begin{equation*}
-\frac52-\rho_2
=
-\frac{3\sqrt{11}}{4}<0.
\end{equation*}
Therefore all second-topic modes are stable. The only unstable closed-loop
mode is the first-topic mode associated with $\operatorname{im}\Pi$, with
eigenvalue $1/2$.

For the affine part of the example, notice that
\begin{equation*}
d=(I_2\otimes A_a)\operatorname{vec}(X_\circ)
=
\operatorname{vec}(X_\circ)
=
\operatorname{col}(b_1,b_2),
\end{equation*}
with
\begin{equation*}
\bar b_1:=\frac1{10}\mathbbm{1}_{10}^{\top}b_1=0.63,
\qquad
\bar b_2:=\frac1{10}\mathbbm{1}_{10}^{\top}b_2=0.
\end{equation*}
Since 
$A_{\mathrm{cl}}$ is nonsingular, the affine compatibility condition
\begin{equation*}
P_\star d+c\in\operatorname{im}(A_{\mathrm{cl}}^\top)
\end{equation*}
holds automatically. Hence there exists a unique affine vector
$p\in\mathbb R^{20}$ satisfying
\begin{equation}
\label{eq:affind-affine-p}
A_{\mathrm{cl}}^\top p+P_\star d+c=0.
\end{equation}
The stationary affine feedback associated with free-endpoint
construction in \cite{trentelman1989regular} is therefore
\begin{equation*}
v^\star(x)=-P_\star x-p.
\end{equation*}
Consequently, the affine closed-loop dynamics are
\begin{equation*}
\dot x
=
A_{\mathrm{cl}}x+d-p.
\end{equation*}

We now compute the structure of $p$. Partition
\begin{equation*}
p=\operatorname{col}(p_1,p_2),
\qquad
x=\operatorname{col}(x_1,x_2),
\qquad
d=\operatorname{col}(b_1,b_2),
\end{equation*}
with $p_1,p_2,x_1,x_2,b_1,b_2\in\mathbb R^{10}$. Since
\begin{equation*}
P_\star d=\operatorname{col}(0,P_2b_2),
\end{equation*}
the first block of \eqref{eq:affind-affine-p} gives
\begin{equation*}
\left(\frac12 I_{10}-L\right)^\top p_1=0.
\end{equation*}
Because
\begin{equation*}
\frac12 I_{10}-L
=
\frac12\Pi-\frac32\Pi_\perp
\end{equation*}
is nonsingular, it follows that
\begin{equation*}
p_1=0.
\end{equation*}
Thus the affine correction has no constant component on the free first-topic
block. The second block of \eqref{eq:affind-affine-p} gives
\begin{equation*}
\left(-\frac12 I_{10}-L-P_2\right)^\top p_2
+
P_2b_2
=
0.
\end{equation*}
Since $L$ and $P_2$ are symmetric, this is equivalently
\begin{equation*}
p_2
=
-
\left(-\frac12 I_{10}-L-P_2\right)^{-1}P_2b_2.
\end{equation*}
Therefore
\begin{equation*}
p=\operatorname{col}(0,p_2).
\end{equation*}

Under the feedback $v^\star(x)=-P_\star x-p$, the closed-loop has exactly one unstable direction, $x_1$, as illustrated in Fig.~\ref{fig:affine_indefinite_modes}. As a final remark, we stress that the complete graph is used only to make the
modal reconstruction explicit.}
\end{example}
\begin{figure}[t]
    \centering
    \includegraphics[width=\linewidth]{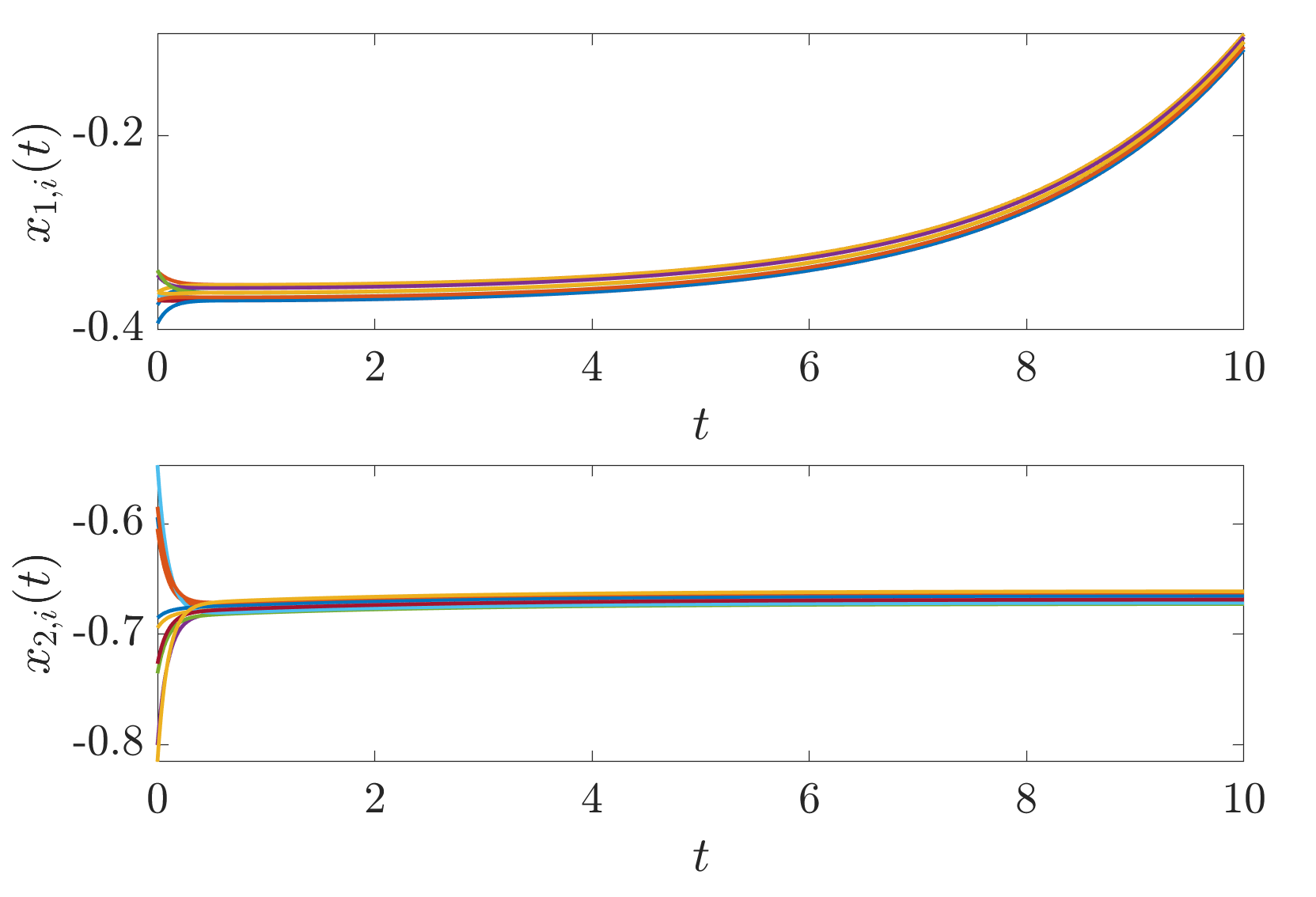}
    \caption{Closed-loop trajectories of the affine indefinite free-endpoint
    solution for the networked two-topic Example~\ref{ex:free_unstable_aff}. The upper panel shows the
    first-topic states $x_{1,i}(t)$, while the lower panel shows the
    second-topic states $x_{2,i}(t)$, for $i=1,\ldots,10$.}
    \label{fig:affine_indefinite_modes}
\end{figure}

\begin{example}[Failure of affine compatibility]
\label{ex:affine_no_control}
{\color{black}
Consider the affine reduced problem \eqref{eq:prob_full_alt}. Let $n=1$, $m=2$, and choose
\begin{equation*}
A_a=1,\qquad
C=
\begin{bmatrix}
1 & \xi\\
0 & \tfrac12
\end{bmatrix},
\qquad
R=I_2,
\end{equation*}
\begin{equation*}
N=\diag(-2,\,-\tfrac52+\eta),
\qquad
Q=\diag\!\Big(4,\;(\tfrac52-\eta)^2-\beta\Big),
\end{equation*}
\begin{equation*}
d=
\begin{bmatrix}
-\kappa/w\\
0
\end{bmatrix},
\qquad
c=
\begin{bmatrix}
\kappa\\
0
\end{bmatrix},
\end{equation*}
with parameters
$0<\eta<\frac52$,
$\xi\in(-1,0)\cup(0,1)$,
$0<\beta<\min\!\left\{\eta^2,\Big(\tfrac52-\eta\Big)^2\right\}$, $\kappa\neq 0$, and
$w>0$.

Then, by \eqref{eq:subs},
\begin{equation*}
\widetilde A=
\begin{bmatrix}
0 & \xi\\[2pt]
0 & -\eta
\end{bmatrix},
\qquad
\widetilde Q=\diag(0,-\beta).
\end{equation*}
Thus the homogeneous reduced stage cost is sign-indefinite, and the first state is completely unpenalized.

The choice of $d$ and $c$ is compatible with the original construction. Indeed,
\begin{equation*}
A_{uc}=C-2I_2=
\begin{bmatrix}
-1 & \xi\\
0 & -\tfrac32
\end{bmatrix},
\end{equation*}
and with
\begin{equation*}
x_{\mathrm{eq}}=
\begin{bmatrix}
-\kappa/w\\
0
\end{bmatrix},
\qquad
W_{\mathrm D}=\diag(w,0),
\end{equation*}
one has
\begin{equation*}
A_{uc}x_{\mathrm{eq}}+d=0,
\qquad
-W_{\mathrm D}x_{\mathrm{eq}}=c.
\end{equation*}

We first analyze the associated homogeneous free-endpoint problem. Let
\begin{equation*}
P=P^\top=
\begin{bmatrix}
P_{11} & P_{12}\\[2pt]
P_{12} & P_{22}
\end{bmatrix}
\end{equation*}
solve the ARE \eqref{eq:ARE-min}. Expanding \eqref{eq:ARE-min} gives
\begin{equation*}
\begin{cases}
-(P_{11}^2+P_{12}^2)=0,\\[2pt]
\xi P_{11}-P_{12}(P_{11}+P_{22}+\eta)=0,\\[2pt]
2\xi P_{12}-2\eta P_{22}-(P_{12}^2+P_{22}^2)-\beta=0.
\end{cases}
\end{equation*}
The first equation implies
\begin{equation*}
P_{11}=0,
\qquad
P_{12}=0.
\end{equation*}
Hence all symmetric ARE solutions are diagonal, and the last equation reduces to
\begin{equation*}
P_{22}^2+2\eta P_{22}+\beta=0.
\end{equation*}
Define
\begin{equation*}
\Delta:=\sqrt{\eta^2-\beta}\in(0,\eta).
\end{equation*}
Then the two symmetric ARE solutions are
\begin{equation*}
P_-=\diag(0,-\eta-\Delta),
\qquad
P_+=\diag(0,-\eta+\Delta).
\end{equation*}

Now
\begin{equation*}
A_-:=\widetilde A-P_-=
\begin{bmatrix}
0 & \xi\\[2pt]
0 & \Delta
\end{bmatrix},
\end{equation*}
so
\begin{equation*}
\ker(P_-)=\mathrm{span}(e_1),
\qquad
A_-e_1=0.
\end{equation*}
Therefore
\begin{equation*}
(\ker P_-\mid A_-)=\mathrm{span}(e_1).
\end{equation*}
Since the eigenvalues of $A_-$ are $0$ and $\Delta>0$, one has
\begin{equation*}
\mathcal X^+(A_-)=\reals^2.
\end{equation*}
Hence, the free subspace is
\begin{equation*}
\mathcal N=(\ker P_-\mid A_-)\cap \mathcal X^+(A_-)=\mathrm{span}(e_1).
\end{equation*}

Because
\begin{equation*}
P_+-P_-=\diag(0,2\Delta),
\end{equation*}
one has
\begin{equation*}
\{x\in\reals^2:\ (P_+-P_-)x\in\mathcal N^\perp\}
=
\mathcal N^\perp=\mathrm{span}(e_2).
\end{equation*}
Hence the corresponding supporting projector is
\begin{equation*}
\Pi_{\mathcal N}=e_1e_1^\top.
\end{equation*}
Therefore the supported free-endpoint solution is
\begin{equation*}
P_\star=P_-\,\Pi_{\mathcal N}+P_+(I-\Pi_{\mathcal N})=P_+.
\end{equation*}
Hence, for the homogeneous problem, the free-endpoint optimal feedback is
\begin{equation*}
v^\star_{\mathrm h}(x)=-P_\star x
=
\begin{bmatrix}
0\\
(\eta-\Delta)x_2
\end{bmatrix}.
\end{equation*}

We now return to the affine problem. Since $P_\star=P_+$, we have
\begin{equation*}
A_{\mathrm{cl}}
=
\widetilde A-P_\star
=
\begin{bmatrix}
0 & \xi\\[2pt]
0 & -\Delta
\end{bmatrix}.
\end{equation*}
In particular,
\begin{equation*}
\ker(A_{\mathrm{cl}})=\mathrm{span}(e_1).
\end{equation*}
Moreover,
\begin{equation*}
g:=P_\star d+c
=
\begin{bmatrix}
0&0\\
0&-\eta+\Delta
\end{bmatrix}
\begin{bmatrix}
-\kappa/w\\
0
\end{bmatrix}
+
\begin{bmatrix}
\kappa\\
0
\end{bmatrix}
=
\begin{bmatrix}
\kappa\\
0
\end{bmatrix}.
\end{equation*}
Therefore
\begin{equation*}
e_1^\top g=\kappa\neq 0.
\end{equation*}
Equivalently,
\begin{equation*}
g\notin \operatorname{im}(A_{\mathrm{cl}}^\top).
\end{equation*}
By Proposition~\ref{prop:affine_consistency}, the bias equation
\begin{equation*}
A_{\mathrm{cl}}^\top p+g=0
\end{equation*}
has no solution, and thus there exists no quadratic-affine stationary identity of the form
\begin{equation*}
V(x)=x^\top P_\star x+2p^\top x
\end{equation*}
associated with a feedback
\begin{equation*}
v^\star(x)=-R^{-1}(P_\star x+p).
\end{equation*}

This example isolates a genuinely affine obstruction in a way that is compatible with the original reduction: the homogeneous free-endpoint problem is algebraically regular and admits a supported optimal feedback, but the affine forcing injects a nonzero component along the free closed-loop direction $e_1$, so the homogeneous controller cannot be extended to a stationary affine controller of the same quadratic-affine form.}
\end{example}

\end{document}